\documentclass[final, 10pt, letterpaper]{IEEEtran}
\usepackage[T1]{fontenc}
\usepackage{url}
\ifCLASSINFOpdf
\else
\fi

\usepackage{amsmath}
\usepackage{algorithmicx}
\usepackage{algorithm}
\usepackage{graphicx}
\usepackage{mathtools}
\usepackage{amsthm}
\usepackage{amssymb}
\usepackage{algpseudocode}
\usepackage{nccmath,textcomp}
\usepackage{eso-pic}
\usepackage{bbm}
\newtheorem{theorem}{Theorem}
\newtheorem{lemma}{Lemma}
\newtheorem{remark}{Remark}
\newtheorem{comment}{Comment}

\newtheorem{proposition}{Proposition}
\newtheorem{assumption}{Assumption}
\newtheorem{corollary}{Corollary}
\newtheorem*{proof-mine}{Proof}
\usepackage{verbatim}

\makeatletter
\let\MYcaption\@makecaption
\makeatother

\usepackage[font=footnotesize]{subcaption}

\makeatletter
\let\@makecaption\MYcaption
\makeatother
\newcommand\blfootnote[1]{%
  \begingroup
  \renewcommand\thefootnote{}\footnote{#1}%
  \addtocounter{footnote}{-1}%
  \endgroup
}
\begin{document}
\title{Analysis of Worst-Case Interference in Underlay Radar-Massive MIMO Spectrum Sharing Scenarios}
\author{Raghunandan M. Rao, Harpreet S. Dhillon, Vuk Marojevic, and Jeffrey H. Reed
\\[-10ex]}
\maketitle
\thispagestyle{empty}
\begin{abstract}
\begin{comment}
In this paper, we consider an underlay radar-massive MIMO spectrum sharing scenario. We model the massive MIMO base station locations using a Poisson point process (PPP), which are allowed to operate outside a \textit{circular exclusion zone} centered at the radar. In this scenario, we study the average interference at the radar due to elevation/3D beamforming in the massive MIMO downlink. 
We devise a novel construction based on the circumradius distribution of a typical Poisson-Voronoi cell, to derive an upper bound on the average interference. 
We also provide a simpler construction based on modeling each cell as a circle with area equal to the average area of the typical cell, and derive tight approximations for these cell-shape models. Finally, we show that the gap in interference power under the two models is approximately constant w.r.t. the exclusion zone radius. Our analysis reveals useful trends in average interference power, as a function of key deployment parameters such as radar/base station (BS) antenna heights, number of elements, BS density and exclusion zone radius.
\end{comment} 
In this paper, we consider an underlay radar-massive MIMO spectrum sharing scenario in which massive MIMO base stations (BSs) are allowed to operate outside a circular exclusion zone centered at the radar. Modeling the locations of the massive MIMO BSs as a homogeneous Poisson point process (PPP), we derive an analytical expression for a tight upper bound on the average interference at the radar due to cellular transmissions. The technical novelty is in bounding the worst-case elevation angle for each massive MIMO BS for which we devise a novel construction based on the circumradius distribution of a typical Poisson-Voronoi (PV) cell. While these worst-case elevation angles are correlated for neighboring BSs due to the structure of the PV tessellation, it does not explicitly appear in our analysis because of our focus on the {\em average} interference. We also provide an estimate of the nominal average interference by approximating each cell as a circle with area equal to the average area of the typical cell. Using these results, we demonstrate that the gap between the two results remains approximately constant with respect to the exclusion zone radius. Our analysis reveals useful trends in average interference power, as a function of key deployment parameters such as radar/BS antenna heights, number of antenna elements per radar/BS, BS density, and exclusion zone radius.
\end{abstract}
\blfootnote{
R. M. Rao, H. S. Dhillon and J. H. Reed are with Wireless@VT, Bradley Department of ECE, Virginia Tech, Blacksburg, VA, 24061, USA (e-mail: \{raghumr,hdhillon,reedjh\}@vt.edu). \\
\indent V. Marojevic is with the Department of ECE at Mississippi State University, Mississippi State, MS, 39762, USA (e-mail: vuk.marojevic@ece.msstate.edu). \\
\indent The support of the U.S. NSF Grants CNS-1564148, CNS-1642873, and ECCS-1731711 is gratefully acknowledged.}

\IEEEpeerreviewmaketitle
\section{Introduction}
Spectrum sharing and massive MIMO are two key spectral efficiency enhancing techniques that have been included in the Third Generation Partnership Project (3GPP) Release 15 specifications. While massive MIMO enhances spectral efficiency by increasing the dimension of spatial multiplexing by an order of magnitude, spectrum sharing improves it by sharing spectrum between different wireless technologies in the spatial and temporal dimensions. Spectrum sharing is particularly attractive in the sub-6 GHz frequency bands, where spectrum is under-utilized due to conservative policies \cite{Sudeep_Abid_DynExcZone_DySPAN_2014}. Among the various incumbents, radars are the biggest consumer of spectrum in the sub-6 GHz bands. In underlay radar-cellular spectrum sharing scenarios where the establishment of an exclusion zone limits cellular interference to the radar, coordination is often impossible due to security concerns, or unfeasible due to practical limitations.
The lack of coordination can potentially exacerbate the interference due to receive and transmit beamforming capabilities of the radar and BS, respectively. Therefore, it is important to understand the worst-case interference at the radar as a function of key deployment parameters in such scenarios, which is the main focus of this paper. 

\subsubsection*{Related Work}
Multi-antenna techniques have been well-explored in the radar-communications coexistence literature. In the case of coordination between the primary and secondary users, MIMO techniques have been investigated in the context of spectrum sharing between a MIMO radar and the MU-MIMO downlink \cite{Liu_Robust_MIMO_BF_Rad_Cell_Coexist_2017}, MIMO radar and full-duplex cellular systems \cite{Biswas_FDMIMO_radar_coexist_TWC_2018}, and MIMO radar and a MIMO communication system \cite{Li_MIMOMC_Radar_TSP_2016}, under performance and power constraints. 
Even though secondary user interference mitigation is possible using multi-antenna radars in uncoordinated scenarios \cite{Deng_Himed_InterfMit_TAES_2013}, its feasibility in the presence of a \textit{large multi-cell network of massive MIMO BSs} is limited to scenarios of sparse deployments and/or large exclusion zone radii. 

Owing to its tractability, tools from stochastic geometry have been used recently to analyze spectrum sharing systems \cite{Li_Baccelli_Andrews_LTE_WiFi_TWC_2016, Parida_Dhillon_CBRS_Access_2017}. Authors in \cite{Hessar_Roy_Radar_WiFi_TAES_2016} considered a radar-WiFi spectrum sharing scenario, where WiFi access points (APs) were modeled as a homogeneous PPP. The exclusion zone radius was computed for different scenarios based on side-information available at the APs. In \cite{Kim_WiFi_Radar_WCL_2017}, the authors evaluated the mean aggregate interference from Wi-Fi APs to radar using tools from stochastic geometry. However, these works consider azimuth-only beamforming, and do not model the impact of elevation beamforming, which is a prominent feature introduced in 5G NR. While \cite{Baianifar_ant_height_mMIMO_PIMRC_2017, Yang_BS_Downtilt_DL_Cellular_TWC_2019} considered the elevation angle, the focus of these works is on antenna height optimization and interference mitigation in cellular networks. 

\subsubsection*{Contributions}
In this work, we develop a novel and tractable analytical framework to analyze the \textit{average interference power} in radar-massive MIMO spectrum sharing scenarios, which is a key metric that has been used in drafting spectrum sharing policies in recent years \cite{FCC_3point5_GHz_Rules}. Incorporating elevation beamforming into the stochastic geometry framework is challenging, since Voronoi cells of the BSs can be arbitrarily large. 
To overcome this, we devise a novel formulation based on the circumradius distribution of the Voronoi cell \cite{Calka_Circumradius_CCDF_2001}. In addition, the presence of sidelobes result in a beamforming gain that is a non-monotonic function of the elevation angle. We derive an upper bound on the beamforming gain that monotonically decreases with the elevation angle, which is crucial to deriving the upper bound on the average interference. We also derive the nominal average interference power by modeling each Voronoi cell as a circle of area equal to the average area of a typical cell. Finally, we provide approximations, that lead to the development of intuitive system design insights regarding the worst-case exclusion zone radius, scaling laws, and the difference between the worst-case and nominal average interference values.
\section{System Model} \label{Sec_Sys_Model}
We consider the radar-massive MIMO spectrum sharing scenario shown in Fig. \ref{Fig1_Radar_mMIMO_SpecShare_Illustration}. The radar is the primary user (PU), equipped with a $N^{(rad)}_{az} \times N^{(rad)}_{el}$ uniform rectangular array (URA) with $\tfrac{\lambda}{2}$-spacing, mounted at a height of $h_{rad}$ m. The massive MIMO downlink is the secondary user (SU), with each BS serving $K$ users with equal power allocation using multi-user MIMO (MU-MIMO). Each BS is equipped with a $N^{(BS)}_{az} \times N^{(BS)}_{el}$ URA with $\tfrac{\lambda}{2}$-spacing, mounted at a height of $h_{BS}$ m. The subscripts $az$ ($el$) are used to denote the azimuth (elevation) elements respectively, and superscripts $rad$ ($BS$) denote the radar (BS) antenna elements respectively. The radar is protected from SU interference by a \textit{circular exclusion zone} of radius $r_\mathtt{exc}$. The exclusion zone is chosen to be circular since there is no coordination between the cellular network and the radar system, and the radar is assumed to search for a target uniformly at random in the azimuth $[-\tfrac{\pi}{2}, \tfrac{\pi}{2})$, as shown in Fig. \ref{Fig1_Radar_mMIMO_SpecShare_Illustration}.\\[-3ex]

\begin{figure}[t]
	\centering
	\includegraphics[width=3.3in]{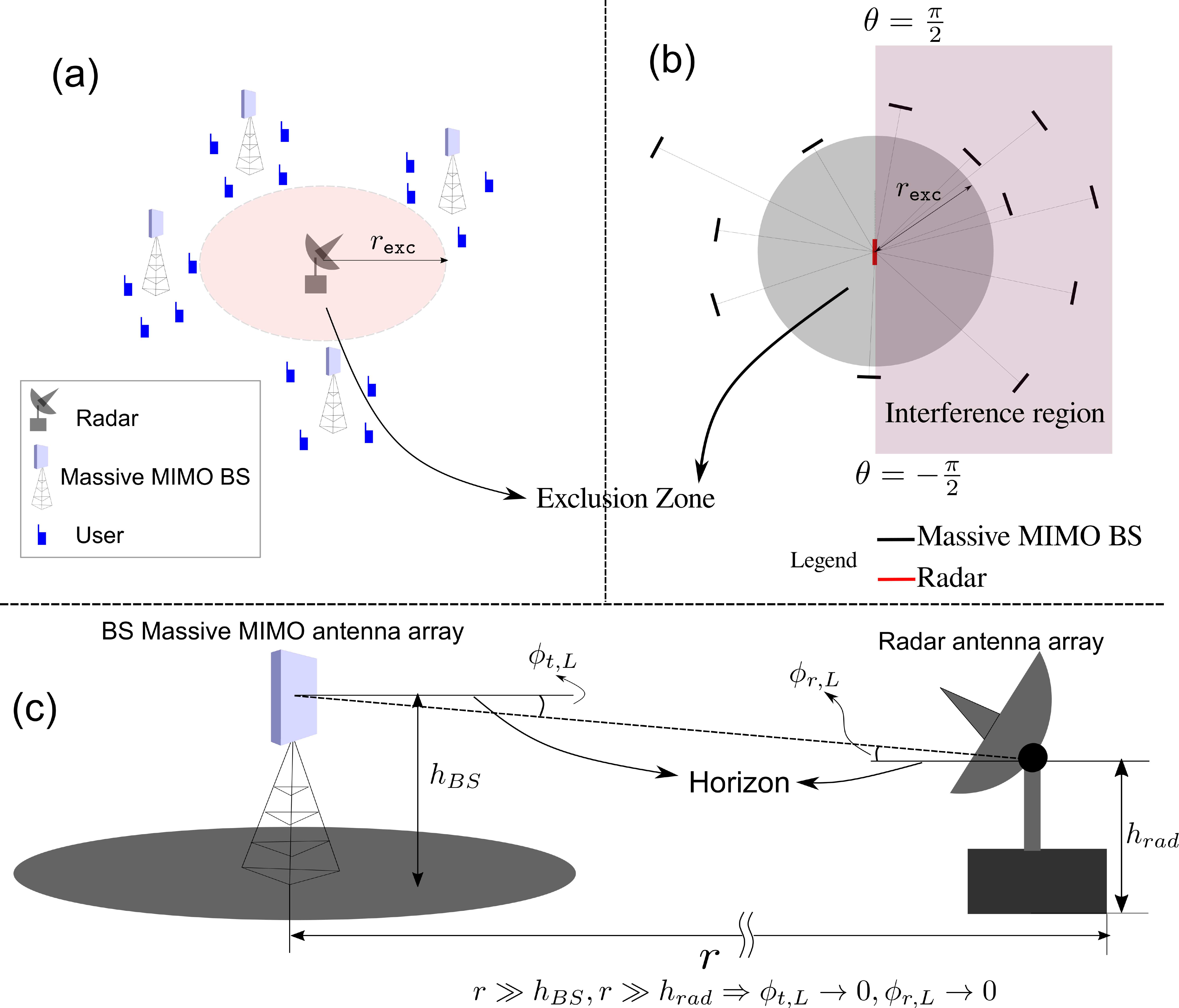}\\
	[-1ex]
	\caption{Illustration of the radar-massive MIMO spectrum sharing scenario, (a) the radar is protected from massive MIMO downlink interference by an exclusion zone of radius $r_\mathtt{exc}$, (b) Top View: the boresight of each BS is aligned along the direction of the radar, and the radar receives interference from the azimuth $\big[\tfrac{-\pi}{2}, \tfrac{\pi}{2} \big)$ depicted by the shaded region, (c) the line of sight component has elevation angle of departure ($\theta_{t,L}$) and arrival ($\theta_{r,L}$) close to $0^\circ$, i.e. the horizon. In our convention, $-\pi/2 \leq \phi < 0^\circ$ for elevation angles above the horizon, and $0 < \phi \leq \pi/2$ for elevation angles below the horizon. \\
	[-6ex]}
	\label{Fig1_Radar_mMIMO_SpecShare_Illustration}
\end{figure}

\subsection{Channel Model}
In quasi-stationary channel conditions, the spatial channel between each BS and the radar is given by \cite{3GPP5GNR_ChanModels}
\begin{align}
	\label{Doub_dir_chan_model_f}
	\mathbf{H_R}(f) = & \sqrt{\tfrac{\beta(d)}{1 + K_R}} \Big( \sqrt{K_R} \mathbf{a} (\theta_{t,L}, \phi_{t,L}) \mathbf{a}^H (\theta_{r,L}, \phi_{r,L}) + \nonumber \\
	& \sqrt{\tfrac{1}{N_c}} \sum\limits_{i=1}^{N_c} \gamma_i \mathbf{a} (\theta_{t,i}, \phi_{t,i}) \mathbf{a}^H (\theta_{r,i}, \phi_{r, i}) \Big),  
\end{align}
where $\beta(d)= PL(r_0) d^{-\alpha}$ is the path loss, $PL(r_0)$ is the path-loss at reference distance $r_0$, $\alpha$ is the path-loss exponent ($\alpha > 2$), $d$ is the 3D distance between the BS and the radar, and $N_c$ is the number of discrete multipath components (MPCs). The Rician factor $K_R\ggg 1$, where propagation is dominated by the line of sight component\footnote{Such propagation scenarios are observed in (a) coastal deployments (where the terrestrial BSs is sharing spectrum with a naval radar), and (b) terrestrial deployments in flat rural/suburban terrain (terrestrial BSs sharing spectrum with terrestrial radar systems).}. In addition, the random small-scale fading amplitude satisfies $\mathbb{E}[\gamma_i] = 0$ and $\mathbb{E}[|\gamma_i|^2] = 1$. 
The azimuth and elevation angles of arrival (departure) of the $i^{th}$ MPC at the radar (from the BS) is denoted by $\theta_{r,i}$ ($\theta_{t,i}$) and $\phi_{r,i}$ ($\phi_{t,i}$) respectively. Similarly, the azimuth and elevation angles of departure (arrival) of the LoS component is given by $\theta_{t,L}$ ($\theta_{r,L}$) and $\phi_{t,L}$ ($\phi_{r,L}$) respectively as shown in Fig. \ref{Fig1_Radar_mMIMO_SpecShare_Illustration}. The steering vector $\mathbf{a}(\theta_t, \phi_t) \in \mathbb{C}^{N^{(BS)}_{az} N^{(BS)}_{el}}$ (BS), and $\mathbf{a}(\theta_r, \phi_r) \in \mathbb{C}^{N^{(rad)}_{az} N^{(rad)}_{el}}$ (radar) is defined in Appendix \ref{App1_Proof_BFGain_UpBound}.\\[-5ex]

\subsection{Massive MIMO Downlink Beamforming Model}
The massive MIMO downlink serves $K$ users located in clusters with mutually disjoint angular support using joint spatial division multiplexing (JSDM) \cite{JSDM_Adhikary_Caire_TIT_2013}. We consider a highly spatially correlated downlink channel, given by the one-ring model \cite{JSDM_Adhikary_Caire_TIT_2013} as $\mathbf{h_i} = \sqrt{\beta_i}\mathbf{U_i} \mathbf{\Lambda}^{1/2}_\mathbf{i} \mathbf{z_i} \in \mathbb{C}^{M}$, where $M = N^{(BS)}_{az} N^{(BS)}_{el} $, $\beta_i$ is the large-scale pathloss for the $i^{th}$ user, $\mathbf{U_i} \in \mathbb{C}^{M \times r}$ is the orthonormal matrix of eigenvectors, $\mathbf{\Lambda_i} \in \mathbb{R}^{r\times r}$ is the diagonal matrix of eigenvalues, and $\mathbf{z_i} \sim \mathcal{CN}(\mathbf{0, I_r}) \in \mathbb{C}^r$ is a complex Gaussian random vector, where $r \ll M$ is the channel rank in the high spatially correlated downlink channel \cite{JSDM_Adhikary_Caire_TIT_2013}. For the sake of simplicity, we consider that all users in the network have the same channel rank. The received signal $\mathbf{y} \in \mathbb{C}^K$ can be written as
\begin{align}
	\label{per_UE_rx_sig}
	\mathbf{y} & = \mathbf{H}^H \mathbf{W_{RF} W_{BB} d}  + \mathbf{n},
\end{align}
where $\mathbf{W_{RF}} = [\mathbf{w_{RF,1}}\ \ \mathbf{w_{RF,2}}\cdots \mathbf{w_{RF,K}}] \in \mathbb{C}^{M\times K}$ is the RF beamformer that groups user clusters with disjoint angular support using nearly orthogonal beams, and $\mathbf{W_{BB}} = [\mathbf{w_{BB,1}}\ \ \mathbf{w_{BB,2}}\cdots \mathbf{w_{BB,K}}] \in \mathbb{C}^{K \times K}$ is the baseband precoder \cite{JSDM_Adhikary_Caire_TIT_2013}. If the azimuth and elevation angular support of the $k^{th}$ user cluster is given by $\Theta_k = [\theta^{(\mathtt{min})}_k, \theta^{(\mathtt{max})}_k]$ and $\Phi_k = [\phi^{(\mathtt{min})}_k, \phi^{(\mathtt{max})}_k]$, then without loss of generality we consider that the RF beamformer is given by $\mathbf{w_{RF,k}} = \tfrac{1}{\sqrt{M}}\mathbf{a}(\theta_k, \phi_k)$, where $\theta_k = (\theta^{(\mathtt{min})}_k + \theta^{(\mathtt{max})}_k)/2$ and $\phi_k = (\phi^{(\mathtt{min})}_k + \phi^{(\mathtt{max})}_k)/2$. The data $\mathbf{d}= [d_1\ d_2\ \cdots\ d_K]^T \in \mathbb{C}^K$, such that $\mathbb{E}[\mathbf{d}]=\mathbf{0}$ and $\mathbb{E}[\mathbf{dd}^H] = \tfrac{P_{BS}}{K} \mathbf{I}$, where $d_k$ is the symbol intended for the $k^{th}$ UE and $P_{BS}$ is the total transmit power per BS. The noise $\mathbf{n} \in \mathbb{C}^{K\times 1}$ is spatially white with $\mathbf{n} \sim \mathcal{CN} (\mathbf{0}, \sigma^2_n \mathbf{I})$.
\begin{proposition}\label{Prop_WBB}
For the massive MIMO BS in the asymptotic regime, the baseband precoding matrix for Zero-Forcing (ZF) and Maximum Ratio Transmission (MRT) can be approximated as $\mathbf{W_{BB}} \approx \mathbf{I}$, when $K$ users from different clusters with mutually disjoint angular support are served. 
\end{proposition}
\begin{proof}
(Sketch) The MRT and ZF precoders are $\mathbf{W}^{(\mathtt{MRT})}_\mathbf{BB} = \mathbf{W}^H_\mathbf{RF} \mathbf{H}$ and
$\mathbf{W}^{(\mathtt{ZF})}_\mathbf{BB} = (\mathbf{H}^H \mathbf{W}_\mathbf{RF})^{-1}$ respectively. In the asymptotic regime $\mathbf{W}^H_\mathbf{RF} \mathbf{W_{RF}} \approx \mathbf{I}$ \cite{JSDM_Adhikary_Caire_TIT_2013}. For users in clusters with mutually disjoint angular support, $\mathbf{U}^H_\mathbf{i} \mathbf{w_{RB,j}} \approx \mathbf{0}, i \neq j$ \cite{JSDM_Adhikary_Caire_TIT_2013}. Therefore, $\mathbf{H}^H \mathbf{W_{RF}} \approx \mathbf{\Upsilon} = \text{diag}[\upsilon_1\ \upsilon_2 \cdots \upsilon_K]$, a diagonal matrix. Since $\mathbb{E}[\mathbf{dd}^H] = \tfrac{P_{BS}}{K}$, when the sum-power constraint $\mathbb{E}[\| \mathbf{W_{RF} W_{BB} d} \|_2] = P_{BS}$ is imposed, we obtain the desired result.
\end{proof}
\begin{remark}
The above is true when $N^{(BS)}_{az}, N^{(BS)}_{az} \rightarrow \infty$. In the case of finite number of antenna elements, we consider a scheduler where the BS co-schedules $K$ users from clusters such that the above approximation is accurate. \\[-4ex]
\end{remark}

\section{Interference at the Radar due to a Single BS}
The radar is assumed to be searching/tracking a target above the horizon ($\phi < 0$) using a receive beamformer $\mathbf{w_{rad}} \in \mathbb{C}^{N^{(rad)}_{az} N^{(rad)}_{el}}$. The received signal prior to beamforming is $\mathbf{y_{rad}} = \mathbf{H}^H_{\mathbf{R}} \mathbf{W_{RF}} \mathbf{W_{BB}} \mathbf{d}$, where $\mathbf{H_R}$ is the high-${K_R}$ Rician channel between the BS and the radar from (\ref{Doub_dir_chan_model_f}). The frequency dependence of $\mathbf{H_R}$ is ignored for the ease of exposition. Upon receive beamforming, the interference signal is given by
\begin{align}
	\label{Int_at_radar_1}
	i_{rad} = \mathbf{w}^H_{\mathbf{rad}} \mathbf{H}^H_{\mathbf{R}} \mathbf{W_{RF} \mathbf{W_{BB}} d}.
\end{align}
Using equation (\ref{Doub_dir_chan_model_f}) in the above and simplifying, we get
\begin{align*}
i_{rad} & = \sqrt{\tfrac{\beta(d)}{K_R + 1}} \Big( \sqrt{K_R G_{rad}(\theta_{r,L}, \phi_{r,L})} e^{-j \alpha_0} \mathbf{a}^H(\theta_{t,L}, \phi_{t,L})+ \nonumber \\
& \sum\nolimits_{i=1}^{N_c} \sqrt{\tfrac{G_{rad}(\theta_{r,i}, \phi_{r,i})}{N_c}} \gamma'_i \mathbf{a}^H(\theta_{t,i}, \phi_{t,i}) \Big) \mathbf{W_{RF} W_{BB}d},
\end{align*}
where $\gamma'_i = \gamma^*_i e^{-j \alpha_i}$, the radar beamforming gain $G_{rad}(\theta_j, \phi_j) = |\mathbf{w}^H_\mathbf{rad} \mathbf{a}(\theta_j, \phi_j)|^2$, and $\alpha_0$ is the residual phase. The specular component can be ignored if $G_{rad}(\theta_{r,L}, \phi_{r,L}) \gg G_{rad}(\theta_{r,i}, \phi_{r,i})$. For a tractable worst-case analysis model, we make the following assumptions.

\begin{assumption}
(LoS beamforming gain dominance) The radar is scanning above the horizon with $\mathbf{w_{rad}} = \tfrac{\mathbf{a}(\theta, \phi)}{\sqrt{N^{(rad)}_{az} N^{(rad)}_{el}}}$ such that $G_{rad}(\theta_{r,L}, \phi_{r,L}) > G_{rad}(\theta_{r,i}, \phi_{r,i})\ \forall\ 1\leq i \leq N_c$.
\end{assumption}
\begin{assumption}\label{BoresightAssumption}
(Boresight assumption) Boresight of the antenna array of each massive MIMO BS is aligned along the direction of radar ($\theta_{t,L}=0$)
as shown in Fig. \ref{Fig1_Radar_mMIMO_SpecShare_Illustration}\footnote{As we will discuss in Appendix \ref{App1_Proof_BFGain_UpBound}, Assumption 2 does not impact the worst-case analysis.}. 
\end{assumption}
\begin{assumption}
The cellular downlink is exactly co-channel with the radar system, and radar and cellular operating bandwidths are equal. Hence, the frequency-dependent rejection (FDR) factor of the radar is unity\footnote{The FDR is dependent on the radar receiver architecture, spectrum of the interfering signal, and is independent of other parameters. The interference power at the radar is inversely proportional to the FDR. Interested readers are referred to \cite{Hessar_Roy_Radar_WiFi_TAES_2016} for more details.}.
\end{assumption}
\begin{assumption}
In each cell, the scheduler allocates resources to users in different clusters, where all but one cluster has disjoint angular support with the boresight of the BS URA.
\end{assumption}
\indent Based on the above assumptions, we have the following lemma.
\begin{lemma}\label{Lemma_DIUC}
(Dominant interfering user cluster) The interference to the radar from each BS is only due to data transmissions towards a single cluster whose angular support overlaps with the boresight of the URA. 
\end{lemma}
\begin{proof}
Let the $K$ clusters have azimuth and elevation angles of support given by $\Theta_k$ and $\Phi_k$ respectively, for $1 \leq k \leq K$. In the asymptotic regime, if there is only one $k$ such that $\Theta_k \cap \{ 0^\circ \} \neq \emptyset$, then we get $\mathbf{a}^H(\theta_{t,L}, \phi_{t,L}) \mathbf{w_{RF,j}} \approx 0$ for $j \neq k$ and $\mathbf{a}^H(\theta_{t,L}, \phi_{t,L}) \mathbf{w_{RF,k}} \neq 0$ \cite{JSDM_Adhikary_Caire_TIT_2013}. The cluster with overlapping angular support is termed as the `Dominant Interfering User Cluster' (DIUC).
\end{proof}
Based the above, we have the following key result. 
\begin{theorem}
The worst-case average interference power at the radar due to the DIUC is given by
\begin{align}
\label{WorstCaseAvgIntPow_SingleBS}
\bar{I}_{rad} < I^{(\mathtt{w})}_{rad} =  \tfrac{\beta(d) G_{rad}(\theta_{r,L}, \phi_{r,L}) |\mathbf{a}^H(0, \phi_{t,L}) \mathbf{a}(\theta_k, \phi_k)|^2 P_{BS}}{N^{(BS)}_{az} N^{(BS)}_{el} K}. 
\end{align}
\end{theorem}
\begin{proof}
Under the realistic assumption that each MPC is uncorrelated with the others, the average interference power $\bar{I}_{rad} = \mathbb{E}[|i_{rad}|^2]$ is given by
\begin{align}
\label{Irad_avgPow}
\bar{I}_{rad} & = \tfrac{\beta(d) K_R G_{rad}(\theta_{r,L} \phi_{r,L}) \mathbb{E}[ \|\mathbf{a}^H(0, \phi_{t,L}) \mathbf{W_{RF} W_{BB} d}\|^2_2]}{K_R + 1} + \nonumber \\
& \sum_{i=1}^{N_c} \tfrac{\beta(d) \mathbb{E}[\gamma'^2_i \| \mathbf{a}^H(\theta_{t,i}, \phi_{t,i}) \mathbf{W_{RF} W_{BB} d}\|^2_2] G_{rad}(\theta_{r,i} \phi_{r,i})}{N_c (K_R + 1)}.
\end{align}
Using Assumption 1, we get $\bar{I}_{rad} < \beta(d) G_{rad}(\theta_{r,L}, \phi_{r,L}) \cdot \mathbb{E}[ \|\mathbf{a}^H(\theta_{t,L}, \phi_{t,L}) \mathbf{W_{RF} W_{BB} d}\|^2_2] $ since $\mathbb{E}[|\gamma'_i|^2]=1$. In addition, by Proposition \ref{Prop_WBB}, Assumption 2 and Lemma \ref{Lemma_DIUC}, we get $\bar{I}_{rad} < \mathbb{E}[|\mathbf{a}^H(0, \phi_{t,L}) \mathbf{w_{RF,k}} d_k|^2] \beta(d) G_{rad}(\theta_{r,L}, \phi_{r,L})$. Finally, using $\mathbb{E}[|d_k|^2]=P_{BS}/K$ and substituting the RF beamformer for the DIUC, we obtain the desired result.
\end{proof}
In summary, the worst-case average interference in high-$K_R$ Rician channels in the asymptotic regime resembles the Friis transmission equation, with the power scaled by the beamforming gains, and the power allocation factor to the DIUC. With this general result, we analyze the average interference due to the cellular network in the next section. 
\section{Analysis of Average Interference at the Radar due to the Massive MIMO DL} \label{Sec_Int_at_Radar_mMIMO_Network}
We model the spatial distribution of the massive MIMO BSs and radars as independent PPPs $\mathbf{\Phi_{BS}}$ and $\mathbf{\Phi_{rad}}$ of intensity $\lambda_{BS}$ and $\lambda_{rad}$ respectively, such that $\lambda_{rad} \lll \lambda_{BS}$. The typical radar is located at the origin, with an exclusion zone of radius $r_\mathtt{exc}$ within which the BSs are prohibited from operating. While the range of azimuth of a randomly selected point in the cell is independent of the cell size, the elevation angle depends on the cell size and hence, on $\lambda_{BS}$. Compared to prior works \cite{Hessar_Roy_Radar_WiFi_TAES_2016}, \cite{Kim_WiFi_Radar_WCL_2017} which focus on beamforming in the azimuth, mathematical modeling of elevation beamforming presents technical challenges due to (a) lack of radial symmetry in the Voronoi cell, (b) possibility of arbitrarily large Voronoi cells, and (c) correlation between adjacent cells, which can affect the elevation distribution.
While correlation between adjacent cells does not deter the analysis since we are interested in the average interference power, the lack of radial symmetry and possibility of arbitrarily large cells need a more thoughtful treatment. In addition, the presence of sidelobes in the beamforming pattern complicates the problem since it is non-trivial to express the worst-case beamforming gain as a function of the cell-size.
Below, we develop the techniques to address these issues, and present the worst-case and nominal average interference analysis. 
\begin{lemma}\label{Lemma_Monotonic_BF_Gain}
(Monotonic beamforming gain function) For the BS $N_{az} \times N_{el}$ URA with $\lambda/2$-spacing, if $\phi \in [-\pi/2, \pi/2), 0 \leq \phi_\mathtt{m} \leq \tfrac{\pi}{2}$, and $\theta \in [-\pi/2, \pi/2)$, then the upper bound of the beamforming gain is given by
\begin{align}
\label{BFGain_Tight_UpperBound}
G^{(\mathtt{max})}_{BS}(\phi, \phi_\mathtt{m}) & = \underset{\substack{\phi_k\in [\phi_\mathtt{m},\pi/2) \\ \theta_k \in [-\pi/2, \pi/2)}}{\max} G_{BS}(\theta, \phi, \theta_k, \phi_k) \\
& = \begin{cases}
N_{az}N_{el}, \quad \quad \quad \quad \text{if } \phi_\mathtt{m} \leq \phi, \\
 G_{BS} (0, \phi, 0, \phi_\mathtt{m}),\ \text{if } \sin \phi_\mathtt{m} \leq \tfrac{1 + N_{el} \sin \phi}{N_{el}} \\
 \tfrac{N_{az}/N_{el}}{\sin^2 \big(\tfrac{\pi (\sin \phi_\mathtt{m} - \sin\phi )}{2}  \big)}, \quad \text{otherwise}
\end{cases} \nonumber 
\end{align}
where $G_{BS}(\theta, \phi, \theta_k, \phi_k) = \tfrac{1}{N_{az} N_{el}}|\mathbf{a}^H(\theta,\phi) \mathbf{a} (\theta_k, \phi_k)|^2$.
\end{lemma}
\begin{proof}
See Appendix \ref{App1_Proof_BFGain_UpBound}.
\end{proof}
\subsection{Circumcircle-based Cell (CBC) Model}
To induce radial symmetry in the Voronoi cell, it needs to be modeled as a circle. In addition, the worst-case interference to the radar occurs when the BS beamforms to the \textit{farthest point in the cell}, according to Lemma \ref{Lemma_Monotonic_BF_Gain}. Since the circumradius determines the distance to the farthest point in a cell, we propose a circumcircle-based construction as shown in Fig. \ref{Fig_Cell_Shape_Models}, with the following PDF. 
\begin{proposition}
The probability density function of the circumradius $r_c$ ($r_c > 0$) of a Poisson-Voronoi cell is 
\begin{align*}
f_{R_C} (r_c) = 8 \pi \lambda_{BS} r_c e^{-4 \pi \lambda_{BS} r^2_c} \Big[1 + \sum\nolimits_{k \geq 1} \Big\{ \tfrac{(-4 \pi \lambda_{BS} r^2_c)^k}{k!} \cdot \nonumber \\ 
\quad\ \  \Big(\tfrac{\psi_k(r_c)}{8 \pi \lambda_{BS} r_c} - \zeta_k (r_c) \Big) - \tfrac{(-4 \pi \lambda_{BS} r^2_c)^{k-1} \zeta_k (r_c) }{(k-1)!}  \Big\} \Big],  \\
\zeta_k(r_c) = \int_{\|\mathbf{u} \|_1 = 1, u_i \in [0,1]} \Big[\prod_{i=1}^{k} F(u_i) \Big] e^{4 \pi \lambda_{BS} r^2_c \sum\limits_{i=1}^k \int\limits_{0}^{u_i} F(t) dt} d\mathbf{u}, \nonumber \\
\psi_k (r) = \tfrac{d \zeta_k(r)}{d r}, 
F(t) = \sin^2(\pi t) \mathbbm{1}(0\leq t \leq \tfrac{1}{2}) + \mathbbm{1}(t > \tfrac{1}{2}), \nonumber
\end{align*}
where $\mathbbm{1}(\cdot)$ denotes the indicator function.
\end{proposition}
\begin{proof}
The result is obtained by differentiating the CDF of the circumradius ($F_{R_C}(r_c)$) \cite{Calka_Circumradius_CCDF_2001} w.r.t. $r_c$ using Leibniz's rule.
\end{proof}

Using $f_{R_C}(r_c)$ and Lemma \ref{Lemma_Monotonic_BF_Gain}, we obtain the upper bound on the average interference in the following key result.
\begin{theorem}\label{Circum_rad_Model}
The worst-case average interference at the radar is given by 
\begin{align}
\label{Boresight_Cell_Edge_BF_CECC}
\bar{I}_{rad, \mathtt{cbc}} & = \tfrac{\lambda_{BS} P_{BS} PL(r_0)}{K} \int_{-\tfrac{\pi}{2}}^{\tfrac{\pi}{2}} \int_{r_{\mathtt{exc}}}^\infty \int_{0}^{\infty} G_{rad}(\theta_{r,L},-\phi_{t,L}(r) ) \cdot \nonumber \\
&  \tfrac{r G^{(\mathtt{max})}_{BS}(\phi_{t,L}(r), \phi_\mathtt{m}(r_c)) }{(r^2 + (h_{rad} - h_{BS})^2 )^{\alpha/2}}  f_{R_C}(r_c) d r_c dr d\theta_{r,L}, \\
\phi_{t,L}(r) & = \tan^{-1} \big( \tfrac{h_{BS} - h_{rad}}{r} \big),
\phi_\mathtt{m}(r_c) = \tan^{-1} \big( \tfrac{h_{BS}} {r_c}\big).\nonumber 
\end{align}
\end{theorem}
\begin{proof}
See Appendix \ref{App2_Proof_Worst_Case_Int}.
\end{proof}
\begin{corollary}\label{Corollary_Circum_Model}
The approximate worst-case average interference at the radar is given by
\begin{align}
\label{Approx_Circum_Model_avgInt}
\bar{I}^{(\mathtt{approx})}_{rad, \mathtt{cbc}} & = \tfrac{\lambda_{BS} P_{BS} PL(r_0) }{K(\alpha - 2) r^{\alpha - 2}_{exc}} \Big[ \int_{-\tfrac{\pi}{2}}^{\tfrac{\pi}{2}}  G_{rad}(\theta_{r,L},0) d\theta_{r,L} \Big] \cdot \nonumber \\
& \ \Big[\int_{0}^{\infty}  G^{(\mathtt{max})}_{BS} (0, \phi_\mathtt{m}(r_c)) f_{R_C} (r_c) dr_c \Big]. 
\end{align} \qedhere
\end{corollary}
\begin{proof}
Since $r \gg h_{BS}$ and $r \gg h_{rad}$, we have $\phi_{t,L}(r) = -\phi_{r,L}(r) \approx 0$, 
and $(r^2 + (h_{BS} - h_{rad})^2)^{\tfrac{\alpha}{2}} \approx r^{\alpha}$. Using these in $\bar{I}_{rad, \mathtt{cbc}}$, grouping the integrands, and integrating over $r$ yields the desired result.
\end{proof}

\begin{figure}[t]
	\centering
	\includegraphics[width=3.4in]{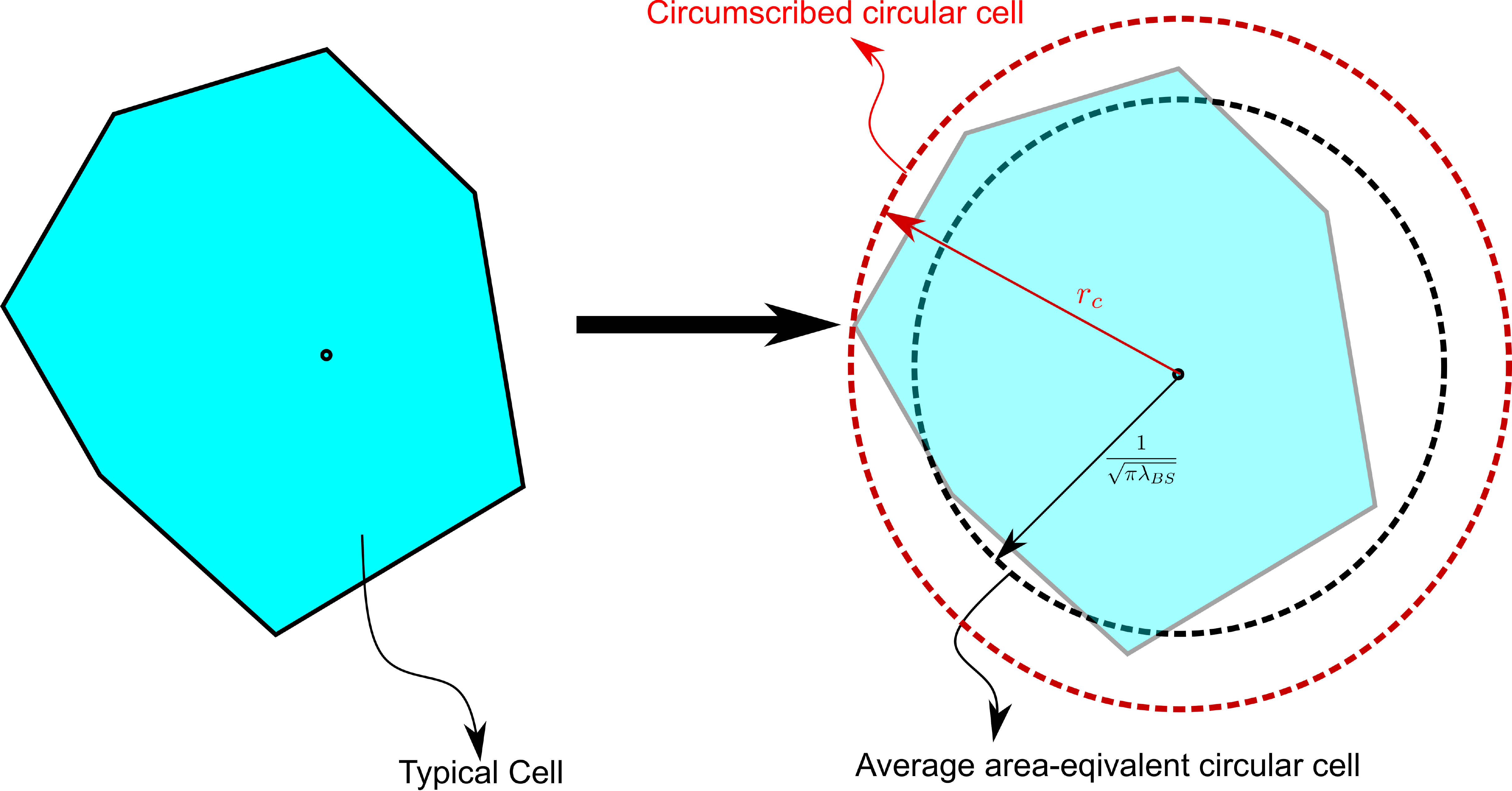}\\
	[-2ex]
	\caption{Radial symmetry can be induced by modeling the Voronoi cell as a (a) circumcircle, or (b) circle of area equal to that of the average typical cell.\\[-10ex]}
	\label{Fig_Cell_Shape_Models}
\end{figure}

\subsection{Average Area-Equivalent Circular Cell (AAECC) Model}\label{AAEC_Model}
The circumcircle-based cell model results in a conservative value for average interference. A simpler, more optimistic model is to replace the Voronoi cell by a circle with an area equal to the average area of a typical cell given by $\tfrac{1}{\lambda_{BS}}$. In this case, the cell radius $r_c = r_a = \tfrac{1}{\sqrt{\pi \lambda_{BS}}}$, and the nominal average interference is given by the following theorem. 
\begin{theorem}\label{Theorem_AACC_Model}
The nominal average interference at the radar is
\begin{align}
\label{Boresight_Cell_Edge_BF_AACC}
\bar{I}_{rad, \mathtt{aaecc}} & = \frac{\lambda_{BS} P_{BS} PL(r_0)}{K} \int_{-\tfrac{\pi}{2}}^{\tfrac{\pi}{2}} \int_{r_{\mathtt{exc}}}^\infty G_{rad}(\theta_{r,L},\phi_{r,L}(r) ) \cdot \nonumber \\
&  \tfrac{r G^{(\mathtt{max})}_{BS}\big(\phi_{t,L}(r),\phi_\mathtt{m}(r_a) \big) }{(r^2 + (h_{rad} - h_{BS})^2 )^{\alpha/2}}  dr d\theta_{r,L}.
\end{align}
\end{theorem}
\begin{proof}
This model is a special case of Theorem \ref{Circum_rad_Model}, where $f_{R_c} (r_c) = \delta \big( r_c - \tfrac{1}{\sqrt{\pi \lambda_{BS}}} \big)$. Using the sifting property of the Dirac delta function $\delta(\cdot)$, we obtain the desired result.
\end{proof}
\begin{corollary}\label{Corollary_AACC}
The approximate nominal average interference is given by
\begin{align*}
\bar{I}^{(\mathtt{approx})}_{rad, \mathtt{aaecc}} & = \tfrac{\lambda_{BS} P_{BS} PL(r_0) G^{(\mathtt{max})}_{BS}\big(0,\phi_\mathtt{m}( r_a ) \big) }{K(\alpha - 2)r^{\alpha - 2}_{exc}}\int_{-\tfrac{\pi}{2}}^{\tfrac{\pi}{2}} G_{rad}(\theta, 0) d\theta.
\end{align*}
\end{corollary}
\begin{proof}
The proof follows the same steps as Corollary \ref{Corollary_Circum_Model}.
\end{proof}

\subsection{System Design Insights from Analytical Results}
\subsubsection{Scaling of average interference power with BS density}
From (\ref{Boresight_Cell_Edge_BF_CECC}) and (\ref{Boresight_Cell_Edge_BF_AACC}), we see that $\lambda_{BS}$ impacts the average interference through the linear, and the BS beamforming gain ($G_{BS}$) terms. It is related to the cell size via the circumradius distribution and the average area of the typical cell, which impacts the \textit{minimum elevation angle} ($\phi_\mathtt{m}$). Note that this dependence is not observed in azimuth-only beamforming models. However, when $h_{BS} \ll r_c$, $\phi_\mathtt{m}(r_c) \rightarrow 0$ and hence, $G_{BS} \rightarrow N^{(BS)}_{az} N^{(BS)}_{el}$. In this regime, the worst-case average interference power scales linearly with $\lambda_{BS}$. 
\subsubsection{Exclusion Zone Radius}
In practice, exclusion zones are defined based on the average aggregate interference power (for e.g. see \cite{FCC_3point5_GHz_Rules}). Using Corollaries \ref{Corollary_Circum_Model} and \ref{Corollary_AACC}, for an average interference threshold $\bar{I}_{th}$, the worst-case exclusion zone radius can be obtained using
\begin{align*}
r^{(\mathtt{wor})}_{exc} & \approx \Big(\tfrac{\lambda_{BS} P_{BS} PL(r_0) }{K(\alpha - 2) \bar{I}_{th}} \Big[ \int_{-\tfrac{\pi}{2}}^{\tfrac{\pi}{2}}  G_{rad}(\theta_{r,L},0) d\theta_{r,L} \Big] \cdot \nonumber \\
& \ \Big[\int_{0}^{\infty}  G^{(\mathtt{max})}_{BS} (0,\phi_\mathtt{m}(r_c)) f_{R_C} (r_c) dr_c \Big] \Big)^{\tfrac{1}{\alpha - 2}}, \alpha >2. 
\end{align*}
\subsubsection{Constant Gap in Average Interference Predicted by CBC and AAECC Models}
By Corollaries (\ref{Corollary_Circum_Model}) and (\ref{Corollary_AACC}), we observe that the ratio of average interference powers is nearly independent of $r_\mathtt{exc}$, given by
\begin{align*}
\eta &= \tfrac{\bar{I}^{(\mathtt{approx})}_{rad, \mathtt{cbc}}}{\bar{I}^{(\mathtt{approx})}_{rad, \mathtt{aaecc}}}= \tfrac{\int\displaylimits_{0}^{\infty}  G^{(\mathtt{max})}_{BS} (0, \phi_\mathtt{m}(r_c)) f_{R_C} (r_c) dr_c}{G^{(\mathtt{max})}_{BS}\big(0, \phi_\mathtt{m}\big( \tfrac{1}{\sqrt{\pi \lambda_{BS}}} \big) \big)}.
\end{align*}
Note that $\eta \rightarrow 1$ when $h_{BS} \sqrt{\pi\lambda_{BS}} \rightarrow 0$ due to BS gain saturation. 

\begin{figure}[t]
	\centering
	\includegraphics[width=3.4in]{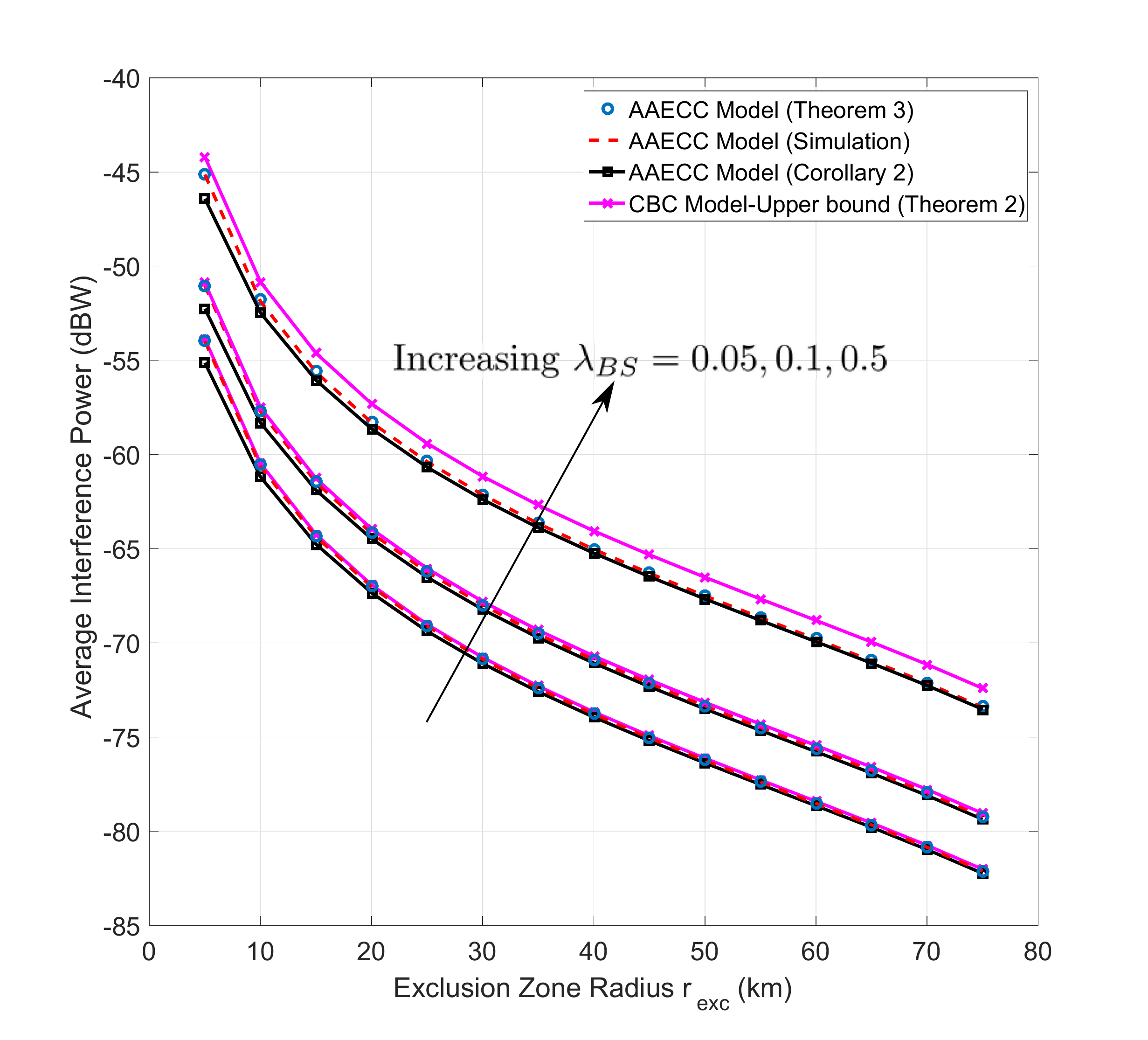}\\
	[-2ex]
	\caption{Worst-case average interference power at the radar due to downlink massive MIMO transmissions, as a function of exclusion zone radius for different $\lambda_{BS}$ (km$^{-2}$). \\[-6ex]}
	\label{Fig_WorstCaseInt_all_Globecom}
\end{figure}

\section{Numerical Results}\label{Sec_Numerical_Results}
In this section, we validate the worst-case interference expressions using Monte-Carlo simulations. We consider a typical radar operating at $f_c = 5$ GHz, located at the origin equipped with a $40 \times 40$ URA, mounted at a height of $h_{rad} = 20$ m. The radar is assumed to be scanning the region at $(\theta, \phi) = (60^\circ, -10^\circ )$ above the horizon. We consider a finite massive MIMO network in a circular region around the origin with a radius of $100$ km. The BSs are distributed as a PPP, with varying intensities. Each massive MIMO BS is co-channel with the radar, and is equipped with a $10 \times 10$ URA deployed at a height of $h_{BS} = 50$ m. The circular exclusion zone around the radar has a \textit{minimum radius} of $r^{(\mathtt{min})}_\mathtt{exc} = 5$ km. The boresight of each massive MIMO BS URA is aligned along the direction of the radar ($\theta_k=0$ in the LCS). 
In each cell, the massive MIMO BS transmits a total power of $P_{BS} = 1$ W, equally allocated among co-scheduled UEs from $K=4$ clusters with mutually disjoint angular support. To model the pathloss in the downlink and the BS to radar channels, we assume the 3GPP 3D Urban Macro (3D UMa) LoS pathloss model \cite{3GPP5GNR_ChanModels}
\begin{align*}
PL(d) = P(h_{BS}, h_{rad}) + 20 \log_{10}(f_c) + 40 \log_{10}(d) \quad \text{(dB)} \nonumber \\
P(h_{BS}, h_{rad}) = 28 - 9 \log_{10}((h_{BS} - h_{rad})^2) \quad \text{(dB)},
\end{align*}
where $f_c$ (GHz), and $d$ (m).

\begin{table}[t]
	\renewcommand{\arraystretch}{1.0}
	\caption{Approximate Values of $\eta$}
	\label{Tab_Ratio_Vals}
	\centering
	\footnotesize
	\begin{tabular}{|c|c|c|c|c|c|c|}
		\hline
		$h_{BS}\sqrt{\pi \lambda_{BS}}$ & 0.0089 & 0.0198 & 0.028 & 0.044 & 0.0886 & 0.1253\\
		\hline
		$\eta$ & 1.004 & 1.022 & 1.045 & 1.254 & 1.608 & 2.905 \\
		\hline
	\end{tabular} \\
	[-3ex]
\end{table}

Fig. \ref{Fig_WorstCaseInt_all_Globecom} plots the average interference power derived in this paper under different cell models, as a function of exclusion zone radius for different BS intensities. We observe that the upper bound is remarkably tight, especially for low values of $\lambda_{BS} \leq 0.1$. For reference, we also plot the approximate average interference power from corollary \ref{Corollary_Circum_Model}. It can be seen that its accuracy improves as $r_\mathtt{exc}$ increases, due to the accuracy of the underlying approximations regarding $d$ and $\phi_{r,L}$. The approximately linear scaling of average interference power with $\lambda_{BS}$ can also be observed.

From Fig. \ref{Fig_WorstCaseInt_all_Globecom}, we observe that the ratio of average interference powers $\eta$ is approximately constant, and is tabulated for the elevation parameter $h_{BS} \sqrt{\pi \lambda_{BS}}$ in Table \ref{Tab_Ratio_Vals}. For typical 3GPP UMa deployment scenarios $h_{BS}/r_{ISD}=0.05$ \cite{3GPP5GNR_ChanModels}, where $r_{ISD}$ is the inter-site distance. The equivalent $h_{BS}\sqrt{\pi \lambda_{BS}}=0.095$, for which $ \eta < 3$ dB. Hence, the upper bound of the average interference power is remarkably tight for typical 3GPP deployment parameters.

\section{Conclusion and Proposed Work}\label{Sec_Conc_Prop_Work}
In this paper, we presented a novel construction based on modeling a Poisson Voronoi cell by its circumcircle, to analyze the worst-case average interference at a typical radar due to a co-channel massive MIMO downlink in a high $K_R$ Rician channel. The proposed model accounted for elevation beamforming capabilities of the massive MIMO BS and the radar, and uncovered the relationship between the BS density and the worst-case BS transmit beamforming gain. We also proposed and analyzed the nominal average interference using an alternate, simpler model, where each cell is replaced by a circle of area equal to the average area of a typical cell. 
Finally, we provided useful insights regarding the worst-case exclusion zone radius, scaling of interference power with BS density, and the approximate gap between the worst-case and nominal average interference power. Our analysis was validated using Monte-Carlo simulations, and we demonstrated that the upper bound using the circumcircle-based model is remarkably tight for realistic deployment parameters. As cellular deployment progresses towards 5G NR with elevation beamforming capabilties in shared spectrum, the analysis techniques presented in this work helps system designers establish important baselines regarding worst-case average interference during simulation and empirical evaluation of radar-cellular coexistence scenarios. \\[-3ex]

\appendix
\subsection{Proof of Lemma \ref{Lemma_Monotonic_BF_Gain}}\label{App1_Proof_BFGain_UpBound}
The steering vector of a $N_{az} \times N_{el}$ URA is $\mathbf{a}(\theta,\phi)=\mathbf{a}_{az}(\theta, \phi) \otimes \mathbf{a}_{el}(\phi)$, where $\otimes$ is the Kronecker product. For $\tfrac{\lambda}{2}$-spacing,
\begin{align*}
\mathbf{a}_{az}(\theta,\phi) & = [1\ e^{-j \pi \sin \theta \cos \phi}\cdots e^{-j \pi (N_{az} - 1)\sin \theta \cos \phi}] \in \mathbb{C}^{N_{az}}, \\
\mathbf{a}_{el}(\phi) & = [1\ e^{-j\pi \sin \phi}\cdots e^{-j\pi (N_{el} - 1)\sin \phi}] \in \mathbb{C}^{N_{el}}.
\end{align*}
Using the properties of the Kronecker product, we get
$G_{BS}(\theta,\phi,\theta_k, \phi_k) =  \tfrac{|\mathbf{a}^H(\theta, \phi) \mathbf{a}(\theta_k, \phi_k)|^2}{N_{az} N_{el}}= \tfrac{|\mathbf{a}^H_{az}(\theta, \phi) \mathbf{a}_{az}(\theta_k, \phi_k) |^2}{N_{az}} \cdot  \tfrac{|\mathbf{a}^H_{el}(\phi) \mathbf{a}_{el}(\phi_k)|^2}{N_{el}}$.
After expanding and simplifying, we get
\begin{align*}
G_{BS}(\theta,\phi,\theta_k, \phi_k) & = \tfrac{\sin^2 \big(\tfrac{\pi}{2} N_{az}(\sin\theta \cos \phi - \sin \theta_k \cos \phi_k) \big)}{N_{az}\sin^2 \big(\tfrac{\pi}{2} (\sin\theta \cos \phi - \sin \theta_k \cos \phi_k) \big)} \times \nonumber \\
& \ \ \ \tfrac{\sin^2 \big(\tfrac{\pi}{2} N_{el}(\sin \phi - \sin \phi_k) \big)}{N_{el}\sin^2 \big(\tfrac{\pi}{2} (\sin \phi - \sin \phi_k) \big)} \leq N_{az} N_{al}.
\end{align*}
Since $\tfrac{\sin^2(Na)}{\sin^2a} \leq N^2 \text{ for } a \in \mathbb{R}$, the universal upper bound is obtained above, and is achieved when $a = 0$. To obtain a tighter bound $G^{(\mathtt{max})}_{BS}$ defined in (\ref{BFGain_Tight_UpperBound}), we consider the following.
\subsubsection{Case 1} If $\phi_\mathtt{m} \leq \phi \leq \tfrac{\pi}{2}$, $G_{BS}(\theta, \phi,\theta_k, \phi_k)$ is maximized by $\phi_k = \phi$, $\theta_k = \theta$, yielding $G^{(\mathtt{max})}_{BS}(\phi, \phi_\mathtt{m}) = N_{az} N_{el}$.

\subsubsection{Case 2}
By upper bounding the \textit{azimuth beamforming gain} in $G_{BS}(\cdot)$, we get $G_{BS}(\theta,\phi,\theta_k, \phi_k) \leq N_{az} \tfrac{\sin^2 \big(\tfrac{\pi}{2} N_{el}(\sin \phi - \sin \phi_k) \big)}{N_{el}\sin^2 \big(\tfrac{\pi}{2} (\sin \phi - \sin \phi_k) \big)}$.
The RHS monotonically decreases w.r.t. $\phi_k$ when $0 \leq \sin \phi_\mathtt{m} \leq \tfrac{1 + N_{el} \sin \phi}{N_{el}} \leq \tfrac{\pi}{2}$ and hence, the upper bound will be given by $G^{(\mathtt{max})}_{BS}(\phi, \phi_\mathtt{m}) = N_{az} \tfrac{\sin^2 \big(\tfrac{\pi}{2} N_{el}(\sin \phi - \sin \phi_\mathtt{m}) \big)}{N_{el}\sin^2 \big(\tfrac{\pi}{2} (\sin \phi - \sin \phi_\mathtt{m}) \big)}$.

\subsubsection{Case 3}
If $\tfrac{1 + N_{el} \sin \phi}{N_{el}} \leq \sin \phi_\mathtt{m}$, the numerator of $G^{(\mathtt{max})}_{BS}(\cdot)$ in case 2 can be upper bounded as $\sin^2(b) \leq 1\  \forall \ b \in \mathbb{R}$, resulting in a monotonically decreasing function of $\phi_\mathtt{m}$. Hence, $G^{(\mathtt{max})}_{BS}(\phi, \phi_\mathtt{m}) = \tfrac{N_{az}}{N_{el}\sin^2 \big(\tfrac{\pi}{2} (\sin \phi - \sin \phi_\mathtt{m}) \big)}$.
\begin{remark}
The upper bound on the beamforming gain is independent of the azimuth angle, since the maximum azimuth beamforming gain can be upper bounded by $N_{az}$. Therefore for the sake of simplicity, we consider that the boresight of each BS is aligned along the direction of the radar, which corresponds to $\theta= 0^\circ$ as discussed in Assumption \ref{BoresightAssumption}. \\[-4ex]
\end{remark}
\subsection{Proof of Theorem \ref{Circum_rad_Model}}\label{App2_Proof_Worst_Case_Int}
Since the radar and massive MIMO BSs are independent PPPs $\mathbf{\Phi_{rad}}$ and $\mathbf{\Phi_{BS}}$ of intensities $\lambda_{rad}$ and $\lambda_{BS}$ respectively with $\lambda_{rad} \lll \lambda_{BS}$, the worst-case average interference at the typical radar is given by Campbell's theorem using
\begin{align*}
\bar{I}_{rad, \mathtt{cbc}} & = \mathbb{E} \Big[ \mathbb{E} \Big[\sum_{\mathbf{X} \in \mathbf{\Phi_{BS}} \setminus \mathbf{\Phi_{exc}} } \{ I^{(\mathtt{w})}_{rad} (\mathbf{X_i}, h_{BS}, h_{rad})|r_c \}\Big]\Big| r_c \Big] \nonumber \\
& = \mathbb{E} \Big[ \int_{\mathbf{x} \in \mathbb{R}^2 \setminus \mathbf{\Phi_{exc}} } \lambda_{BS}\{ I^{(\mathtt{w})}_{rad} (\mathbf{x}, h_{BS}, h_{rad})|r_c \} d \mathbf{x} \Big| r_c \Big],
\end{align*}
where $\mathbf{x} = [r\cos \theta_{r,L}\ r\sin \theta_{r,L}]$, $\mathbf{\Phi_{exc}} = \{r|r \leq r_\mathtt{exc} \}$ denotes the circular exclusion zone, and $r_c$ is the cell radius that determines $G^{(\mathtt{max})}_{BS}(\phi, \phi_\mathtt{m})$ in equation (\ref{BFGain_Tight_UpperBound}). 
Substituting (\ref{WorstCaseAvgIntPow_SingleBS}) above, noting that  
$\phi_{r,L}(r)=-\phi_{t,L}(r)=\tan^{-1} \big( \tfrac{h_{rad} - h_{BS}}{r} \big)$, and converting to polar coordinates we get 
\begin{align}
\bar{I}_{rad,\mathtt{cbc}} & = \mathbb{E} \Big[ \int_{r_\mathtt{exc}}^{\infty} \int_{-\tfrac{\pi}{2}}^{\tfrac{\pi}{2}} \lambda_{BS} \beta(d) G_{rad}(\theta_{r,L}, \phi_{r,L}(r)) \cdot \nonumber \\
& \ G^{(\mathtt{max})}_{BS}(\phi_{t,L}(r), \phi_\mathtt{m}(r_c)) \tfrac{ P_{BS}}{K}  r dr d\theta_{r,L} \Big| r_c \Big],
\end{align}
where $d = \sqrt{r^2 + (h_{BS} - h_{rad})^2}$, and $\beta(d) = PL(r_0) d^{-\alpha}$ is the pathloss model. Using these and integrating over $r_c \sim f_{R_c} (r_c)$, we get the desired result. 
\bibliographystyle{IEEEtran}
\bibliography{prelim_references}
\ifCLASSOPTIONcaptionsoff
  \newpage
\fi
\end{document}